\documentclass{article}
\usepackage[utf8]{inputenc}
\usepackage[english]{babel}

\usepackage{amsthm}

\newtheorem{theorem}{Theorem}[section]

\newtheorem{lemma}[theorem]{Lemma}

\usepackage{changepage}
\usepackage{setspace}
\usepackage[english]{babel}
\usepackage[centertags]{amsmath}
\usepackage{amssymb}
\usepackage{amsfonts}
\usepackage{newlfont}
\usepackage{graphicx}
\usepackage{amsmath}
\usepackage{mathtools}
\usepackage{textcomp}
\usepackage[OT1,T1]{fontenc}
\usepackage{authblk}
\usepackage{braket}
\usepackage{booktabs}
\usepackage{geometry}
\usepackage{xcolor}

\newcommand{\RN}[1]{\textup{\uppercase\expandafter{\romannumeral#1}}}

\usepackage{xcolor}

\usepackage{nameref}
\usepackage{ tipa }

\usepackage{algorithm}
\usepackage{algorithmic}
\usepackage{hyperref}
\usepackage{subfig}

\begin{document}

\title{{\sc\Large Verifiable Computing}\\ Using Computation Fingerprints Within FHE\\
{\rm\large (Preliminary Version)}}

\author{Shlomi Dolev $~~~$ Arseni Kalma\\
Department of Computer Science, Ben-Gurion University of the Negev, Israel}

\maketitle

\begin{abstract}
We suggest using Fully Homomorphic Encryption (FHE) to be used, not only to keep the privacy of information but also, to verify computations with no additional significant overhead, using only part of the variables length for verification. This method supports the addition of encrypted values as well as multiplication of encrypted values by the addition of their logarithmic representations and is based on a separation between hardware functionalities. The computer/server performs blackbox additions and is based on the separation of server/device/hardware, such as the enclave, that may deal with additions of logarithmic values and exponentiation.

The main idea is to restrict the computer operations and to use part of the variable for computation verification (\textit{computation fingerprints}) and the other for the actual calculation. The verification part holds the FHE value, of which the calculated result is known (either due to computing locally once or from previous verified computations) and will be checked against the returned FHE value. We prove that a server with bit computation granularity can return consistent encrypted wrong results even when the public key is not provided. For the case of computer word granularity the verification and the actual calculation parts are separated, the verification part (the consecutive bits from the LSB to the MSB of the variables) is fixed across all input vectors. We also consider the case of Single Instruction Multiple Data (SIMD) where the computation fingerprints index in the input vectors is fixed across all vectors.

\end{abstract}


\setcounter{page}{1}
\pagenumbering{arabic}
\section{Introduction}
The fundamental problem of
\textit{how can a delegator verify that the delegatees performed the computation correctly, without running the computation itself?} formed in \cite{badrinarayanan2018non} and arose as early as 1992.
Important works, as \cite{kilian1992note} or \cite{micali1994cs}, advanced the field of computation delegation by expanding current theoretical models, such as random oracle or standard assumption models, but did not offer a practical enough solution. Other general methods, such as the ones presented in \cite{dolev2013probabilistic} consider hardware design of arithmetic circuits. This problem was also stated explicitly in \cite{goldwasser2015delegating}, ``Such a mechanism of general secure computation outsourcing was recently shown to be feasible in theory, but to design mechanisms that are practically efficient remains a very challenging problem''. \cite{goldreich1991proofs} Stated that zero-knowledge proofs do exist outside the domain of cryptography and number theory, using no assumptions.
One of the only practical breakthroughs, that resulted in a real working application, is the SNARK \cite{bitansky2017hunting}. Although it is designed to be practical, and having an extremely efficient verification procedure, the complexity of creating the verification itself is poly-logarithmic to the length of the original computation and may nullify the benefit in delegating the computation. In addition to that, it only works on known (cleartext) input while typically delegation of computing is based on encrypted inputs and FHE.

Several other approaches were suggested in \cite{FGP14,GGP10,GW13,CKV10,AIK10,CC+18}, which verify computations with schemes involving the usage of \emph{authentication tags}, \emph{signatures} or different \emph{MACs}, however none suggest to couple the result with the indication in a monolithic fashion. A separate indication on a correct computation does not imply that another instance of computation is also correct. Other approaches are based on check-able traces of the computation such as MPC-in-the-head \cite{ishai2007zero} or the probabilistic check-able proofs (PCP) \cite{AS98},
where several randomly (or blindly) chosen steps in the MPC/PCP can be checked to validate the computation. These approaches either requires space overhead and/or interactions with the servers. In some scenarios a user cannot change the (e.g., the server is in a cloud computing scenario) system when an indication on wrong computation is received. Moreover, the servers may decide to act maliciously in burst, so repetition may not assist in learning the number of repetitions needed, let alone the overhead needed in repeating the computation. Moreover, we suggest a more efficient approach by coupling the computation values with the verification section in the same unit of computation, while mantaining a ``pre-processing'' model \cite{applebaum2010secrecy, chung2010improved, gennaro2010non} which can work in a transitive manner, and in proportion of $\log_2$ to the number of inputs.

Fully Homomorphic Encryption (FHE) \cite{gentry2009fully, rivest1978data} enables computation of arbitrary functions on encrypted data without knowing the secret key. Many of the FHE schemes (e.g., \cite{bos2013improved, brakerski2012fully, brakerski2014leveled, brakerski2014efficient, brakerski2011fully, cheon2015fully, coron2014scale, van2010fully, doroz2016homomorphic, ducas2015fhew, gentry2012homomorphic, gentry2013homomorphic, lopez2012fly, cheon2017homomorphic}) followed Gentry's suggested blueprint \cite{gentry2009fully}. FHE becomes more and more practical, and in many scenarios include the evaluation of various calculation or algorithms on encrypted and sensitive data \cite{cheon2015homomorphic, kim2017secure, lauter2014private, naehrig2011can, wang2016healer}. FHE was implemented by several open-source libraries, where the libraries we focus on are Microsoft's SEAL \cite{sealcrypto} and IBM HELib \cite{halevi2014algorithms}, each implementing different schemes from the ones stated above, where
we are interested in BGV \cite{brakerski2014leveled} and CKKS \cite{cheon2017homomorphic} schemes. Unfortunately, while FHE copes with honest but curious servers, it is not designed to cope with malicious/Byzantine \cite{lamport2019byzantine} servers. Note that in the cloud the identity of the server that executes the delegated computing task is not known, and therefore there can be no binding of a server to (wrong) results, and later re-execution may yield the same no binding results, thus, malicious servers can try their luck forever. We propose to add computation fingerprints, that are analogous to error detection codes and can detect deviation from the requested computation. Those computation fingerprints will be computed exactly once by the delegator, by executing the requested computation on some random base values, where those values may act as a witness for future calculations of the same circuit/program/procedure. The result of this calculation will be used, as the same base values will be coupled to all future other computation requests, and the result will be accepted if and only if the fingerprint result was received from the server as previously calculated.
Along with that, we give real implementation examples of usable functions and survey possible attack models on our proposed schemes.

We start with the impossibility result, showing that an adversary can return arbitrary predefined results encrypted with the unknown key used to encrypt the FHE inputs, blindly and consistently, using various techniques. Later on, we consider restricted computation capabilities, where blackbox procedures are used while keeping the computation capabilities to compute any arithmetic circuit over the inputs. In Section \ref{Fingerprints in Addition in Word Computation Granularity}, we present a scheme for adding \textit{computation fingerprints} inside a FHE number,  represented as encrypted bits, allowing only additions. To verify the computations we use unary or integers, along with various restrictions applied to the blackbox procedures to make our scheme safe. Following that, in Section \ref{Multiplication via Logarithms in Word Representation} we present an additional technique, using logarithmic representations, supporting also multiplication, thus, supporting the computation of complex computation as any polynomial and therefore any arithmetic circuit. Here we also consider that the input is represented as an encrypted value (word), and not as encrypted bits like in the previous section. In Section \ref{Integration and Implementation in Word Granularity} we survey an implementation example using existing FHE libraries, different number representation implications, and managing the work of values under different fields. Lastly, in Section \ref{Restrictive SIMD as a Vector with Fingerprints} we survey operations, supporting computation fingerprints that indicate whether the required program/computation is performed on the inputs. Where unlike in the previous section the inputs can be vectors used with SIMD operations, the inputs are not restricted to integers (represented using a computer word), and can be general values including floating point numbers. In Section \ref{Conclusion Remarks} we complete with our conclusions.
\section{Consistent Wrong FHE Results in Bit Granularity}
\label{Consistent wrong FHE results in bit granularity}
We present an impossibility result for detecting wrong computation by repeating the computing possibly using different (FHE) keys and comparing the decrypted results when the computer is allowed to perform bitwise operations. We start with a small scale introduction (\textit{blind conditioning}) and later present a scheme for the server to produce consistent wrong results to the same function. One important ingredient in our demonstration is the execution of blind conditioning.
\\
\\
\noindent{\bf Blind conditioning.}
\label{methods:blind}
One may assume that encryption prevents some possible conditioned computation, we demonstrate the opposite. We present an important feature possible on homomorphically encrypted data, which is the ability to perform {\em blind if} on it. This does not allow us to read the input as plaintext, rather, it lets us run calculations and conditions blindly on (encrypted) data. Namely, we demonstrate blind execution of a basic condition:  if a specific bit $b$ is true (say, represented as an encrypted 1) output $f(x)$, otherwise, output $g(x)$. We plan to use the bit $b$ and calculate the following $b  \cdot f(x)+(enc(1)-b ) \cdot g(x)$.
The implementation is done as described in Algorithm \ref{alg:Blind}, by having a specific $bit$ as the conditioned value, a function $f(x)$ for the positive case, and another function $g(x)$ for the negative case (extracting an $enc(1)$ value will be demonstrated in the next section).

\begin{algorithm}
\caption{Implementation of blind conditioning by a specific bit}
\label{alg:Blind}
\begin{algorithmic}
\REQUIRE $enc(bit), f(x), g(x)$
\STATE return $(enc(bit) \cdot f(x)+(enc(1)-enc(bit)) \cdot g(x)))$
\end{algorithmic}
\end{algorithm}

The above manipulation is based on bit representation of the encrypted value, namely, a vector representation of the data. This allows access to single bits, which offer greater flexibility in doing blind conditioning. Having access to an encrypted bit lets us blindly ``figure'' the bit value and act accordingly, at the cost of several multiplications and an addition of the argument built earlier. It is important to note again that we are not able to ``see'' the actual clear text bit value, but able to work on it and output a result dependent on it. The \textit{blind conditioning} will be later extended to an even more elaborate technique that is also capable of implementing blind switch/case programming primitives (see Algorithm \ref{alg:LUT}), rather than a single blind if.

At first glance it seems possible to repeat a computation encrypted with different keys, and find out whether the server is computing correctly, however, as we show next, the server can be answering consistently wrong answers. The ability to send consistent output can be trivially demonstrated by a policy of the server in which the server uses the input as the output. Such a policy may obviously be suspected by the computation delegator. The server may employ more sophisticated consistent wrong computations, such as, constantly adding (or constantly subtracting) as long as the operation is different from the requested computation, the input variables (in case the input consists of more than one operation).

The computation delegating party may have means to check the result. For example, checking the value of the least significant bit, to reflect the bit anticipated value. One may also occasionally compute the result to compare with the result the server sends, this still gives non-neglected probability for using hunchbacked wrong results, and in the scope of cloud computing does not necessarily reveal the malicious server \cite{shaer2018b}. Other self-testing self-correcting techniques \cite{dolev2013probabilistic, blum1993self} can be used to verify the result. The malicious server may design and tailor a function that nullifies the benefits of (such) easy (easier than the actual calculation) attempts to check techniques. We next prove that the server can implement {\em any} function on the encrypted inputs and be consistent with the answers, repeating the same outputs to the same corresponding (before encryption) inputs, and also be tailored to be consistent across outputs.
 
We now prove that the server can execute any (wrong) function, and still, possibly return consistent outputs across inputs, coping with self-testing/correcting checks, on any encrypted input, and output an encrypted answer for the chosen function.
Given an encrypted input, the server produces an arbitrary encrypted (wrong) computation result of its choice. This result is encrypted with the same key used to encrypt the input, while the server has no access to the encryption (public) key itself. The bootstrapping procedure \cite{gentry2009fully} requires the public key for refreshing the ciphertext to allow unlimited number of operations, still our result concerning the ability to return consistent wrong output is stronger as we do not supply the public key to the server, thus, also address more limited HE schemes; schemes that are constrained to a certain computation depth.

To ensure that the choice of the wrong result is consistent with future queries with the same input, possibly encrypted by different key(s), the server may (explicitly or implicitly) construct in its memory a (clear-text) lookup table, mapping inputs to (specific wrong) outputs.

Every time an encrypted input is received, the server employs a procedure to return an encrypted output with the unknown key used on the encrypted input, using its predefined lookup table, even without having the public key which was used for the encryption.
This means that the server can save a particular and consistent result as cleartext in its lookup table, resulting in the same selected operation for all future inputs, even if the encryption key is repeatedly changed.

The implication is that the server can, after selecting a defined output for every possible input in its lookup table, return consistently manipulated results on all future inputs.
\\
\\
\noindent{\bf Implementation details.}
\label{Consistent:Implementation}
We are assuming that a malicious server has only a single encrypted input (which can be in fact a concatenation of several inputs), for which it wants to return a certain output, according to the input, described in its clear-text Look-Up-Table (LUT).
\begin{figure*}[h]
\centering
\includegraphics[width=3.5cm]{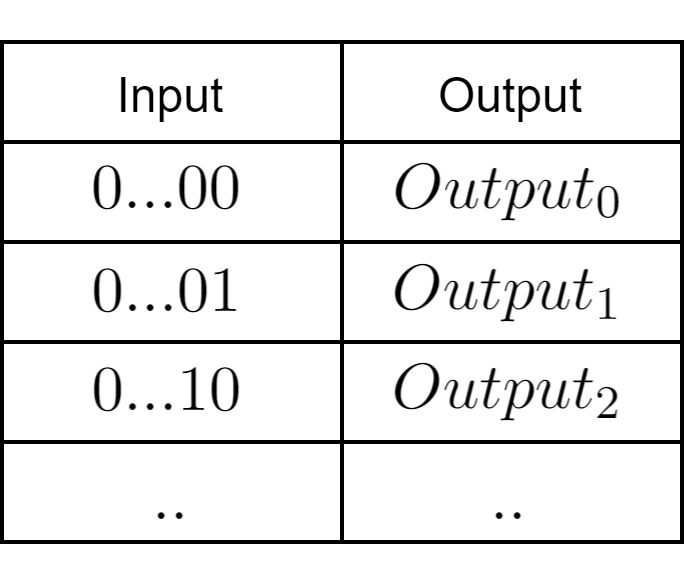}
\caption{Lookup table example}
\label{fig:LUT}
\end{figure*}

For the first step, we show how a malicious server obtains an encrypted value 1 when needed, for the use of setting the selected outputs from its plaintext Look Up Table (LUT) to an encrypted form, without having the public encryption key. Note that, obviously, the avoidance in supplying the public key is limiting the manipulation power of the server, but still allows any computation, as the delegator may supply upfront all needed encrypted values as part of the inputs. Our technique uses solely FHE supported functions, multiplication and addition, to create a bitwise OR operation between two arbitrary encrypted values (bits). The implementation is based on the bitwise OR properties, which can be represented as $\textrm{OR}(x, y) = (x+y) - (x \cdot y)$, meaning that an OR operator will be implemented as a procedure with a single addition, multiplication, and subtraction.
This OR operator is used upon all the bits of the encrypted input iteratively, resulting in an encrypted value 1, unless all of the encrypted input bits are 0. Obtaining an encrypted 0 value is much easier, by simply subtracting any selected input bit by itself. The case in which the input value is exactly 0 is addressed by the fact that when no encrypted 1 can be constructed, the output must be 0.

After having encrypted 1 and 0 values, the malicious server can use regular FHE operations to construct its desired encrypted output values by the received input, even without the public key. We can also negate specific bits in the encrypted input, by subtracting them from our constructed encrypted 1 value, using $Negation(enc(bit))=enc(1)-enc(bit)$.

The next step will be the construction of a calculation that allows us to blindly find the correct row, in terms of the received encrypted input, and return our desired encrypted output. This is made blindly by forming an argument returning a true (1) value that is dependent on the input, for each possible row in the LUT. For every row in the LUT, the argument will consist of multiplication of every bit, or negated ({\it neg}) bit of the input, in a way that the multiplications will equal 1 if and only if the input matches the row input value. This makes the argument equal to 1 if and only if the input has the corresponding bits set to true or false according to the row input value, and to 0 otherwise. By multiplying the expected output at the end of the argument, we set the row value accodingly, if it was the correlated row input value. Thus, we (blindly) get the output of the row if the inputs match the row and 0 otherwise.
\\
\\
For example, for the first 2 rows in the LUT, we construct the following arguments: \\ \\
\smallskip
$neg(bit_n) \cdot \ldots \cdot neg(bit_2) \cdot neg(bit_1) \cdot neg(bit_0) \cdot Output_0 $ \\
$neg(bit_n) \cdot \ldots \cdot neg(bit_2) \cdot  neg(bit_1) \cdot enc(bit_0) \cdot Output_1 $ \\ \\
And the last row will be set as follows:\\ \\
$enc(bit_n) \cdot \ldots \cdot enc(bit_2) \cdot enc(bit_1) \cdot enc(bit_0) \cdot Output_n $ \\ \\
After calculating the value of all the possible rows, we sum the results, where we have zeros for all the rows, except for at most one row, which represents the received encrypted input, and that row calculation will result in our expected output corresponding to the input. It is important to note that we do not break any of the FHE properties, and are not able to read the encrypted input as cleartext, but constructing the blind (and computation heavy) argument, allows us to return a defined output for every possible input, using solely the encrypted input from the delegator.\\
A basic example defined and implemented using HELib of an input range of three bits, where the output for each possible input (0-7) is set alternately to 0 and 1 can be found in \cite{github}.
\begin{algorithm}[H]
\caption{Implementation of bitwise OR operator}
\label{alg:OR}
\begin{algorithmic}
\REQUIRE $x, y$
\STATE return $((x + y) - (x \cdot y))$
\end{algorithmic}
\end{algorithm}

\begin{algorithm}[H]
\caption{Implementation of the Negation operator}
\label{alg:Negation}
\begin{algorithmic}
\REQUIRE $enc(value), enc(1)$
\STATE return $(enc(1)-enc(value))$
\end{algorithmic}
\end{algorithm}

\begin{algorithm}[H]
\caption{Use $LUT$ $L$ for an encrypted input $x$}
\label{alg:LUT}
\begin{algorithmic}
\REQUIRE $enc(x_0), enc(x_1), \ldots , enc(x_n), LUT \ L, \ enc(1)$
\FOR{$i=0 \ to \ n$}
\STATE $negated_i = Negation(enc(x_i), enc(1))$
\ENDFOR
\STATE $out[0] = negated_n \cdot \ldots \cdot negated_2 \cdot negated_1 \cdot negated_0 \cdot output_0$
\STATE $out[1] = negated_n \cdot \ldots \cdot negated_2 \cdot negated_1 \cdot enc(x_0) \cdot output_1$
\\ \ldots
\STATE $out[n] = enc(x_n) \cdot \ldots \cdot enc(x_2) \cdot enc(x_1) \cdot enc(x_0) \cdot output_n$
\FOR{$i=0 \ to \ n$}
\STATE $result = result + out[i]$
\ENDFOR
\STATE return $result$
\end{algorithmic}
\end{algorithm}

The conclusion from all the above is that a malicious server can arbitrarily manipulate the outputs in a consistent way across FHE values encrypted with different keys. Thus, we better restrict the computation primitives the server can use.

Next, we examine several possible restrictions, in particular, restricting the computer to execute only the addition operation. We still investigate the bit computation granularity and examine the binary addition in the process of inputs additions.
\section{Fingerprints in Addition in Word Computation Granularity}
\label{Fingerprints in Addition in Word Computation Granularity}
\noindent
A computation following our scheme will be traceable for verification purposes by adding encrypted input control bits, called {\em computation fingerprints}, inside a single fully homomorphic encrypted number \cite{gentry2009fully}, where the use of subtraction is handled as the addition of negative values. After understanding the proposed solution we survey and test its limits analyzing the probability for successful detection of a malicious server's wrong calculation, where the soundness of our scheme and the {\em computation fingerprint} verification will be presented and surveyed later on.
\\
\\
\noindent
{\bf Arithmetic circuit representation.}
\label{Arithmetic circuits}
We first design our (word) schemes using arithmetic circuits, that represent polynomials (possibly over a certain finite field). Arithmetic circuits are defined as the multiplication or addition of two other variables or constants and provide a formal way to represent the polynomial computing complexity, of which we concentrate in our scheme. In the case of the addition of two polynomials, the resulted polynomial will have the same number of monomials, with component-wise addition, where the multiplication of polynomials will result in a long number of single monomials, added to each other. This representation and specifically the multiplication composition will increase (even exponentially) the size of the polynomial description, yet will be set within our scheme limitation as it will result in a (large) number of separated monomials, calculated by addition of (also, logarithmic represented values for facilitating multiplications) values, and added to each other.
\\
\\
\noindent {\bf Restricting to blackbox additions.}
\label{Restricting to blackbox}
One of the basic operations in computation is an addition, recalling that arithmetic circuits are a combination of additions and multiplications. We examine the server capabilities to manipulate the addition operation when the computation is executed over bit representation. The ability to replace the carry bits computed in each bitwise operation by 0, will lead to the result of 00 when adding 01 to 01. In other words, the server can just ignore or set to zero any carry bit that arose from the addition operation, resulting in a wrong result.

Thus, we further restrict the primitives of the server to execute blacbox word granularity additions, where the server calls a function that returns the correct addition result rather than involving a binary addition.

This means that from now on, we focus on building a computation blackbox supporting environment, similar to the interpreter level (e.g., Java byte-code \cite{venners1998java, gosling2000java}) or restricted operating system as Internet cafe (e.g., \cite{isaacson2009techniques}). The schemes we present can implement any arithmetic circuit computation, while providing trust for the correctness of the computed calculation.
\\
\\
\noindent {\bf Adversarial capabilities of blackbox addition restricted server.}
\label{Adversarial capabilities}
We restrict the operations of the (possibly malicious) server to execute blackbox word granularity additions, while the server is still able to deviate from performing exactly one addition of each of the inputs to be added. This means that the use of blackboxs (analogously to the use of hardware machine opcode primitives) are always assumed, but still an undesired  sequence of their operation can lead to wrong output.
Note that the blackbox restriction prevents the server from manipulating the encrypted inputs otherwise.
\\
\\
\noindent {\bf Binary fingerprints.}
\label{Binary fingerprints}
We suggest computation (addition) fingerprints to ensure the server added the required variables, exactly once. We construct and work with numbers that are composed of parts that fulfill two purposes; the first is to compute the calculation request itself, and the other part(s) for verifying the computation that was made, which we call {\em computation fingerprint}. All numbers will be segmented respectively, to support the correct addition of parts. The expected result of the fingerprint is known and can be computed by the delegator, and can be reused/reconstructed over different inputs, even working with different servers, as it is encrypted and its exact position is unknown to any server. The final result from the server is verified, and deemed valid or not, based on the computed value of the {\em computation fingerprint}.

For example, if we would like to calculate $001+110$, we add four computation fingerprint bits in the tail of $001$ and $110$ to have, say, $0010010$ and $1100100$,
and check whether the result ends with $0110$ prior to accepting the $111$ of the three most significant bits as the correct result. Note, that the choice of having the binary fingerprints to be located as the least significant bits is motivated by the automatic overflow processing of the actual calculation part, avoiding a mix between the fingerprints and the actual calculation portion. This is our motivation to choose the portion of the {\em computation fingerprint} to reside in the least significant part, as the delegator selects the values comprising the input and accordingly knows the expected result of the fingerprint. Note, that the malicious server may use fingerprint overflow to ruin the result while keeping the fingerprint correctness, in particular, the server can add an addend $2^m+1$ times, where $m$ is the number of fingerprint bits. Thus, our addition blackbox is designed to set the result of addition to zero if there is an overflow from the fingerprint part.
\begin{figure*}[h]
\centering
\includegraphics[width=5cm]{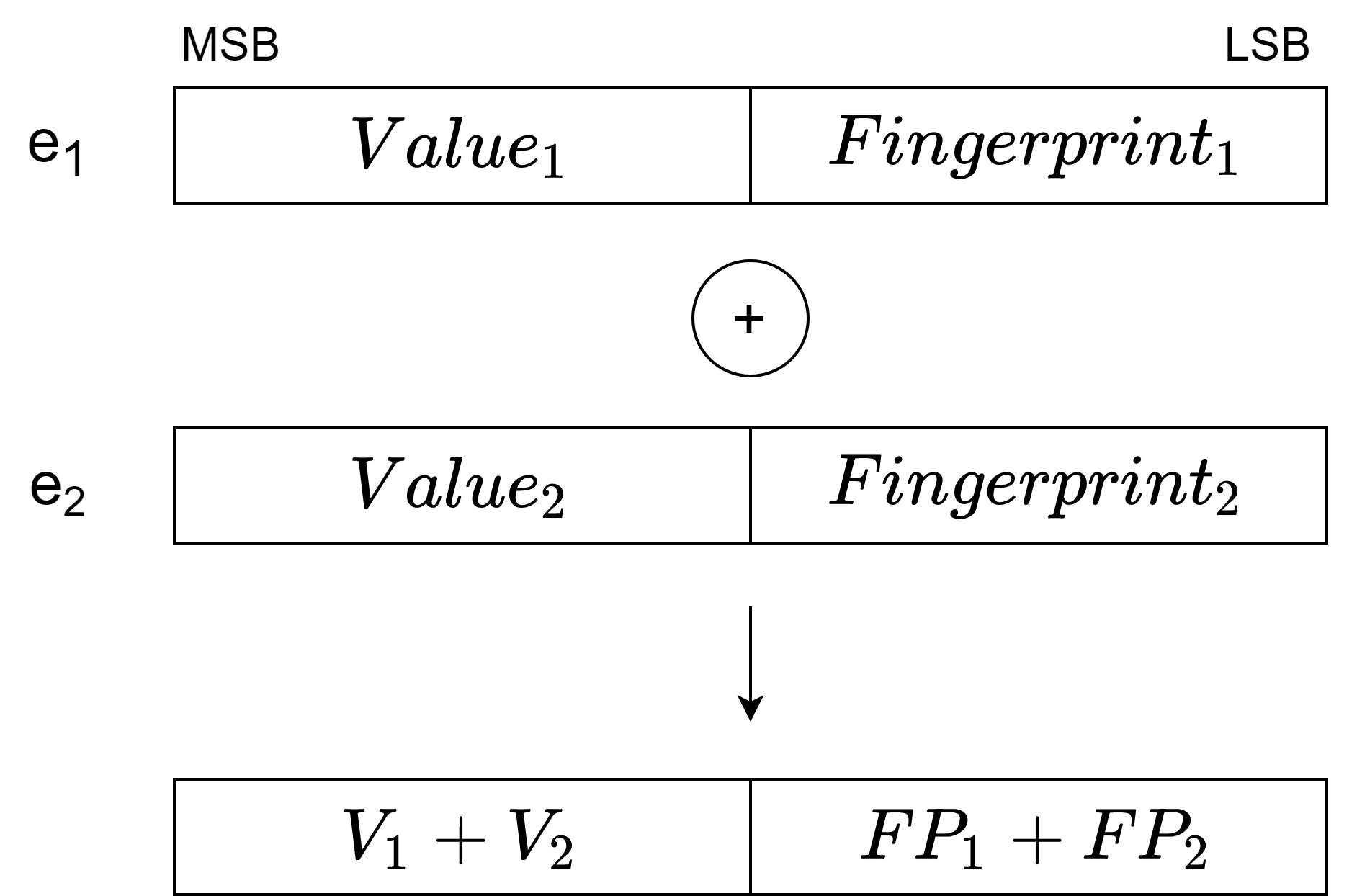}
\caption{Encrypted to encrypted addition}
\label{fig:encrypted to encrypted addition}
\end{figure*}
\\
\\
\noindent {\bf Fingerprint manipulation.}
\label{Fingerprint manipulation}
Obviously, if all inputs are added up at most once and at least one input is added less than once there will be a wrong (missing) fingerprint detected in the result. Alternatively, the server may choose to add an input several times to create the {\em computation fingerprint} in a different manner. For now, we keep each input with a single fingerprint bit in a location different than any other addend for analysis and to prevent overflows from the fingerprint portion through to the calculation portion when the computation is correct. Next, we demonstrate that the adversary can try to guess the locations of the fingerprints and create fake fingerprints by adding an encrypted input more than once, without necessarily triggering an overflow from the fingerprint portion.

Assume that the server knows, or correctly guesses (as the input is encrypted), that there is an encrypted input $x=Enc(0000001)$, then the (malicious) server can produce any fingerprint by adding $x$ to itself, for example, $y=x+x=Enc(0000010)$, and $z=y+y=Enc(0000100)$. Moreover, the server can add $y+z$ to obtain the combination $Enc(0000110)$ of fingerprints.

The (malicious) server may add all addends but one, $x$, and try to create the missing fingerprint of the excluded addend by using the fingerprint of another addend $y$.
The server may succeed in creating the missing fingerprint if the fingerprint of 
$y$ is smaller than the fingerprint of $x$, 
for this, there is $1/2$ probability of guessing correctly, and also guessing the right number of bits $d$ that separate the fingerprint of $x$ from the fingerprint of $y$, for this, the probability of a correct guess is $1/(m-1)$, where $m$ is the number of bits used for the fingerprints. If the guesses are right, which occurs with probability $1/(2m-2)$, then adding $y$ $2^d$ times to the sum that excluded $x$ will mask the absence of $x$.

We also note that, the computation fingerprints are probabilistic in nature, and the probability is related to the number of bits used for the computation versus the number of bits used for the fingerprints. For the sake of convenience we chose these number of bits to be equal, allowing the usage of a past verified computation to serve as fingerprints.
\\
\\
\noindent {\bf Complete fingerprints.}
\label{Complete fingerprints} While defining binary fingerprints and surveying the possibilities to manipulate them, we suggested keeping each input a single fingerprint bit, in a unique location. We strengthen this suggestion and require that the fingerprint bit length correlate to the number of addends used, maintaining the previous suggestion that a single, uniquely located fingerprint bit be used per input, so that the complete addition of inputs will result in a $1..11$ fingerprint. The influence of even a single redundant addition will be a definite overflow of the fingerprint.

Despite the fact that \textit{complete fingerprints} allow the evident indication of excessive additions, as adding all inputs and any of the inputs more than once will create a fingerprint value different than $1..11$; still, a weakness exists as we describe next.
\\
\\
\noindent {\bf Fingerprint overflowing.}
\label{Fingerprint overflowing}
In addition to the adversary possibilities surveyed earlier, we consider a different technique that can allow an adversary to create a calculation result with the expected fingerprint, while still producing a wrong calculation output.

Assuming the fingerprint value resides in the least significant part of the input, the adversary can add an input to itself $2^m$ times to clear out the fingerprint ($m$ is the number of fingerprint bits), while changing the computation value, where an extra addition of the same input will result in the inputs original fingerprint, with a different (corrupted) calculation section. The use of this corrupted input in the requested calculation by the server would create an output with the expected fingerprint value, and an arbitrary wrong calculation result.

Using \textit{complete fingerprints}, an overflow will be immediately identified by the most significant bit in the fingerprint section. This bit can be used with the previously constructed bitwise NOT operator (algorithm \ref{alg:Negation}), calculated on the overflow indicating bit, when at least one of the bits is (encrypted) 1 (when all are zeros, the fingerprint stays zero, which enables detection), multiplied by the given result. If no overflow occurred, this will result in the multiplication of the result by 1, which will not change it. This defines the final blackbox configuration, allowing only addition no overflow-carry operations.
\\
\\
\noindent {\bf Fingerprint carry detection.}
\label{Fingerprint carry detection}
In the \textit{complete fingerprint} scenario, as all inputs have a unique fingerprint bit location, no carry operation in the fingerprint section while adding is required, implying that any carry operation indicates an unintended (malicious) behavior. This behavior might represent an attempt to compensate for a missing input and will influence our blackbox definition, so it nullifies the result when a carry operation, in the relevant section, is detected.

We implement it using the OR and NOT bitwise operators (algorithms \ref{alg:OR}, \ref{alg:Negation}), where the blackbox will OR all the carries that resulted in the fingerprint section, negate this result, and multiply it by the addition result to zero the final output when an undesired carry takes place. In the completing case, where no carries were made, the OR operation will result in a 0, then negated to 1, which will not change or impact the output. For this, the blackbox should be aware of the length $m$ of the fingerprint section, and should be specifically built for that fingerprint length. It is important to note that every input, whether it is the same value or not, is regarded independently and is associated with a unique fingerprint.
\\
\\
\noindent {\bf Capabilities bound of a restricted adversarial.}
\label{Tight bound on restricted adversarial}
The restrictions mentioned result in a tight bound on the adversarial capabilities following our scheme, using the addition no overflow-carry blackbox. We start with the deficient case, where a single missing addend will be obviously detected in the resultant fingerprint. The opposite behavior is the excessive addition of inputs, where a single extra addition will cause a definite overflow in the \textit{complete fingerprints} scheme, which will nullify the whole result. The completing case for all of the above, addressed using the carry detection, will strictly nullify the addition result when any different than intended inputs are added. A scenario where the received input equals 0, is managed by using a \textit{blind if}, and returning the same (0) output. An important note is that the encryption keys are changed in every different request, or use, of the server.

This sets a tight bound on the adversarial capabilities, so we can state the following lemma.

\begin{lemma}
\label{lemma:addition}
For every group of FHE input of bits, having $m$ least significant bits used as complete fingerprints, and the other $n$ bits used for computation, operated on by the addition no overflow-carry blackbox defined earlier, will result in the correct addition of the computation values and a \textit{complete fingerprint}, or the value 0.

\end{lemma}
\begin{proof}
The proof follows from the definition of the addition no overflow-carry blackbox. Let the FHE values $Input_1$, $Input_2$, \ldots  $Input_i$, following our scheme, having complete fingerprint values, operate in the mentioned blackbox. Execution of the blackbox for any of the same inputs will result in a carry, which will be immediately nullified. Omitting a certain input from the calculation will result in a not \textit{complete fingerprint}, and an attempt to compensate for it will involve the carry operation, which will nullify the result in such an event. Note that even if the fingerprint overflowing vulnerability surveyed earlier was not mitigated by the carry detection mechanism, it will be explicitly nullified at any fingerprint overflow.
\end{proof}

\noindent {\bf Integer fingerprints.}
\label{Integer fingerprints}
After considering binary fingerprints, we turn to consider integer fingerprints, where a random number $k$, bigger than one (as 0/1 values are the identity values for addition/multiplication) is chosen uniformly from the range $1$ to $2^m-1$, such that if the addition operation consists of $i$ addends, then the number of overflow bits can be $\log i$ and thus, $k+\log i$ must be less than or equal to $2^m$ ($m$ is the number of fingerprint bits). This new scheme changes the fingerprint value after each addition by more than a single bit, in contrast to the binary fingerprints proposal. This prevents us from nullifying fingerprint carry operations (as it is now needed), yet adds more possibilities for a potential attacker, and lowers its probability to guess the right fingerprint to $1/2^k$. This still requires us to take into account potential overflows, as the adversary can cause an overflow by adding an input more than once. To overcome this, we require the same overflowing mitigation presented earlier. Obviously, the bootstrap of integer fingerprint requires computation, namely, addition of the first fingerprint (for the needed number of addends), later on, the result of the calculations (when no overflow is possible, as enough leading bit values are $0$) may serve as future fingerprint values, based on the correction of the first computed fingerprints yielding the correction of the following computations in a transitive manner.
\\
\\
\noindent {\bf Different inputs subset probability.}
\label{Different inputs subset}
Integer fingerprints expand substantially the range of possible fingerprint values after each addition, yet they might be maliciously used in a new manner. We can represent the verification process we proposed as a variant of the subset sum problem. The set $Z$ of $i$ positive integers are the fingerprints in each of the input values, and the target value $t$ is the expected fingerprint result $2^{k}$. We have a slight variation on the original problem, as the values (inputs) can be repeated. The arithmetic circuit (polynomial) which was requested to be calculated, represents the intended subset that achieves the target value, yet there might be other subsets that result in the same target. The subset sum problem was recently \cite{bringmann2017near}  proved to be solved in pseudo-polynomial time of $\widetilde{O}$($t+i$). Despite that, our scheme has a major difference, as the values used are encrypted, making our problem much more complex, also in the average case, namely, the \textit{blind subset sum} problem, where the server does not know the actual values of its inputs, nor the target value. This leads us to the conclusion that for a calculation with $i$ addends, and a target value $2^k$, any value of the $k$ least significant bits of $m$ have a uniform probability to be chosen and the sum$\pmod{2^k}\ $has also uniform probability in the range of $0$ to $2^{k}-1$, thus the probability of using a (repeated) subset to gain the needed result is less than $1/2^{k}$.
\\
\\
\noindent {\bf Multiplication by a constant.}
\label{Multiplication by a constant}
Taking into account multiplication by a non encrypted constant, the multiplication can be considered as a shortcut of listing certain inputs several times as addends.
More general multiplication can mix the fingerprint part with the actual computation variable, and therefore we suggest using discrete logarithm representation, as we describe next.
The multiplication complexity occurs due to the influence of the bits in both parts in an unexpected manner in the same or completing parts, and their mutual dependency on each other in the encrypted number. This mixture of bits prevents us from accounting for it in our operations, preventing any option to compensate for it in the fingerprint.
\begin{figure*}[h]
\centering
\includegraphics[width=5cm]{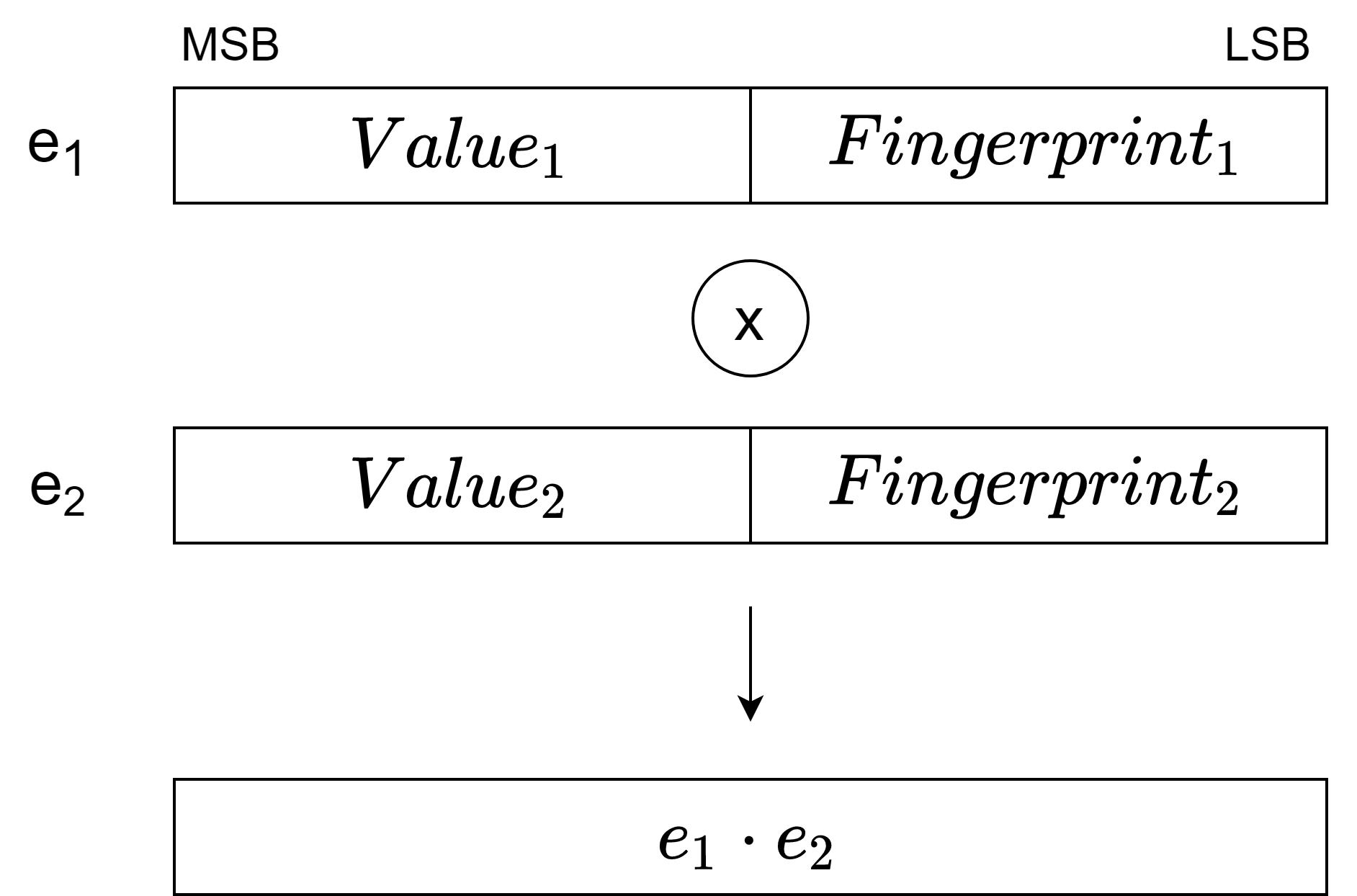}
\caption{Encrypted numbers multiplication}
\label{fig:Encrypted numbers multiplications}
\end{figure*}

The loss of the multiplication operation on encrypted numbers limits and produces a weaker scheme and prevents it from fully implementing arithmetic circuits, which are composites of both additions and multiplications. We evaluate a solution to the deficiency of multiplication in the next sections.
\section{Multiplication via Logarithms in Word Representation}
\label{Multiplication via Logarithms in Word Representation}
\noindent
We suggest using our addition techniques to verify multiplications, namely, use discrete logarithmic representation with fingerprints as inputs.

As our goal is to delegate the complete computation of an arithmetic circuit, we need to verify both additions and multiplications. One possibility is to regard an arithmetic circuit as a polynomial, multiplying to compute monomials, and adding the results of the monomials to complete the polynomial computation.

As a first possibility, the computing server may use LUT to convert the inputs to discrete logarithmic values and to exponentiate the result back before adding.

Then, we propose a possibility to delegate these tasks to a server/hardware/device (e.g., enclaved), or to further enforce that the server compute first multiplications, then use a function that exponentiates the result, while it preserves the encryption and the fingerprints of the inputs, adding an encrypted counter as we detail in the next section.
\\
\\
\noindent {\bf Adversarial capabilities.}
\label{Adversarial capabilites}
We have divided the multiplication process into two phases, the logarithmic addition of values, and the exponentiation of the result. The former was evaluated in the fingerprint manipulation section, in terms of adversarial capabilities, while the latter LUT phase involves the use of an encrypted result for each possible logarithmic addition outcome. The options open to an adversary with an encrypted LUT are limited, where an adversary can only change the order, or amount, of LUT operations on the input(s). This means that any atomic use of the LUT can not be interfered with, yet can be repeatedly abused. Any input other than the intended one, used by the LUT, will result in a wrong fingerprint output being sent to the delegator, or will be wrongly propagated during following operations. An absence of exponentiation will result in a wrong (missing) fingerprint by the server. Any manipulation on the (encrypted) fingerprints values will have the same adversarial probabilities considered in the fingerprint manipulation section.
\\
\\
\noindent {\bf Discrete logarithmic representation.}
\label{Discrete logarithmic representation}
To imitate multiplication between encrypted numbers, we use logarithmic addition in base $2$, which is based on the equation $\log(x \cdot y) = \log(x) + \log(y)$. The solution will be set by the delegator, which is completely transparent to the server. The delegator, instead of requesting multiplication of values, will send the (encrypted) logarithmic representation of the values to the server, which are obviously protected with fingerprints. This makes the server compute additions of log values, which are sent to the delegator as a result to be verified against the fingerprints and then to be exponentiated.

A basic case of computing $4 \cdot 8$, will be represented as $2+3$ in the $\log_2$ field by the delegator. Having the server compute this addition in the proposed scheme will result in the value 5, where the delegator will calculate the result of $2^5$ by itself, and achieve the value $32$, which is the original result of the requested multiplication $4 \cdot 8$.

By using logarithmic addition we can mimic multiplication in our scheme, yet this will introduce other, smaller, limitations.
First, this requires the delegator to calculate the logarithmic representation of the values, possibly by using LUT or caching the already computed logarithms for later reuse.

For example, in the function we show later as $F(x,y) = (2\cdot x)+y+3$ (all constants are encrypted), the calculation of the addition of constants will influence differently than intended, as those values should be reduced to represent their correct significance in the new field. This restricts the complexity of addition and multiplications calculations of individual factors, and form a calculation characterized by a long addition of many monomials.
Moreover, not all log values will result in a round value, restricting to the use of a floating-point FHE scheme, while introducing rounding errors, where it is possible to operate on integers only in additions.
\\
\\
\noindent
{\bf Logarithmic lookup table.}
A possible way to overcome the dependency on the delegator is to exponentiate logarithmic represented values by a blackbox procedure for exponentiation, while keeping the original fingerprint and adding fingerprints for counting the number of exponentiated results that were added.

This means that along with the ordinary fingerprints for additions ($FP_{a}$) which were explained earlier, there will be other (multiplication) fingerprint bits ($FP_{m}$) assigned to track and verify the number of LUT uses made by the server, in the same manner binary fingerprints were used to verify the execution of the addition operator on the input(s). The new fingerprint serves as an indication to avoid the possibility that the server adds the logarithmic representation with their fingerprints, rather than enforcing the exponentiation with fingerprints to ensure the avoidance of this scenario as well.

\begin{figure*}[h]
\centering
\includegraphics[width=3.5cm]{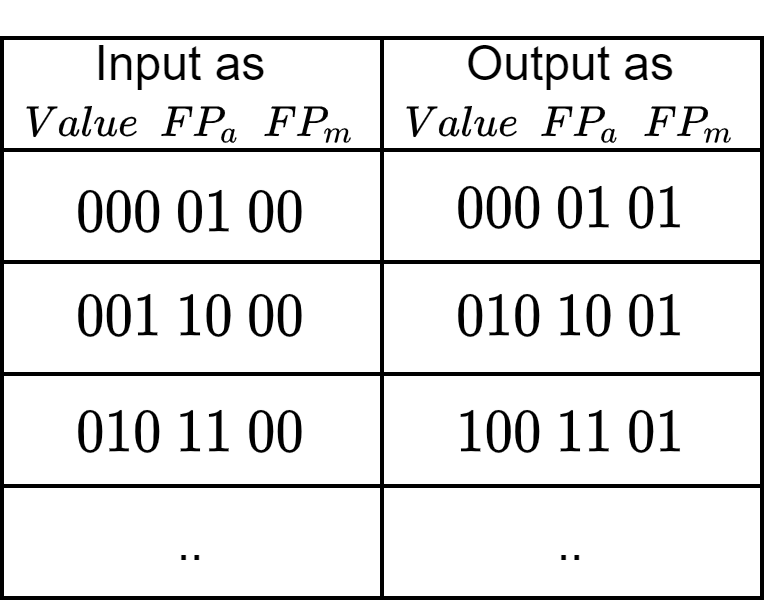}
\caption{Exponentiation lookup table example}
\label{fig:LUT}
\end{figure*}

In continuation from the previous example of $4 \cdot 8$, the server will calculate $2+3$, resulting in $5$, which is represented as $000101$. The fingerprint is of equal size of six bits, let's say, $011000$. For readability, we separate the computation and fingerprint parts with a space. The delegator will set a row in the LUT, for the input of $000101\ 011000$, the result of $100000\ 011001$. This result is composed of the calculation result $100000\ (32)$  and the fingerprint $011001$. The fingerprint was changed from the original $011000$, where the LSB was turned on. This bit is used to indicate in the final result received from the server that this specific exponentiation operation was indeed applied in the calculation process, as the fingerprint was changed accordingly. 

This solution can make the server independent from the delegator and will allow it to continue calculating without interruptions, yet with a computational cost of using the LUT when switching from multiplication to additions.

For example, we calculate the polynomial (that represents an arithmetic circuit), say, $F(x,y) = 2\cdot(x\cdot y+32)$ for $x=4, y=8$. Following the previous example, we expand the fingerprint and calculation parts to eight bits. The requested $4 \cdot 8$, was changed by $log_2$ to $2+3$, resulting in $00000101\ 00011000$.

Using the exponentiation LUT, the value was changed to $00100000\ 00011001$, containing the requested value $32$ ($00100000$), with the fingerprint value $00011001$. The next operation is an addition of 32, having the fingerprint value $00100000$. This sum is represented as $01000000\ 00111001$.

Next, we request to multiply by 2, which is represented as the addition of the value 1. Adding the value $1$ will obligate transforming $64$ to $6$ by the same $\log_2$, while retaining the fingerprint value. This will be achieved by using another LUT, changing the value $01000000\ 00111001$ to $00000110\ 00111010$, where the counting fingerprint value was increased by 1. Finally, we will add the value $1$ with it's fingerprint $01000000$ to result $00000111\ 01111010$.
By using the exponentiation LUT again, we will change the computation part, which equals $7$, to $128$, and increase the fingerprint by 1. This yields the final output of $10000000\ 01111011$. The fingerprint value is known and consistent for the selected fingerprint values, and will be compared and verified when acquired by the delegator. The whole calculation process is shown in Figure \ref{fig:Calculation LUT example}.

\begin{figure*}[h]
\centering
\includegraphics[ height=5cm]{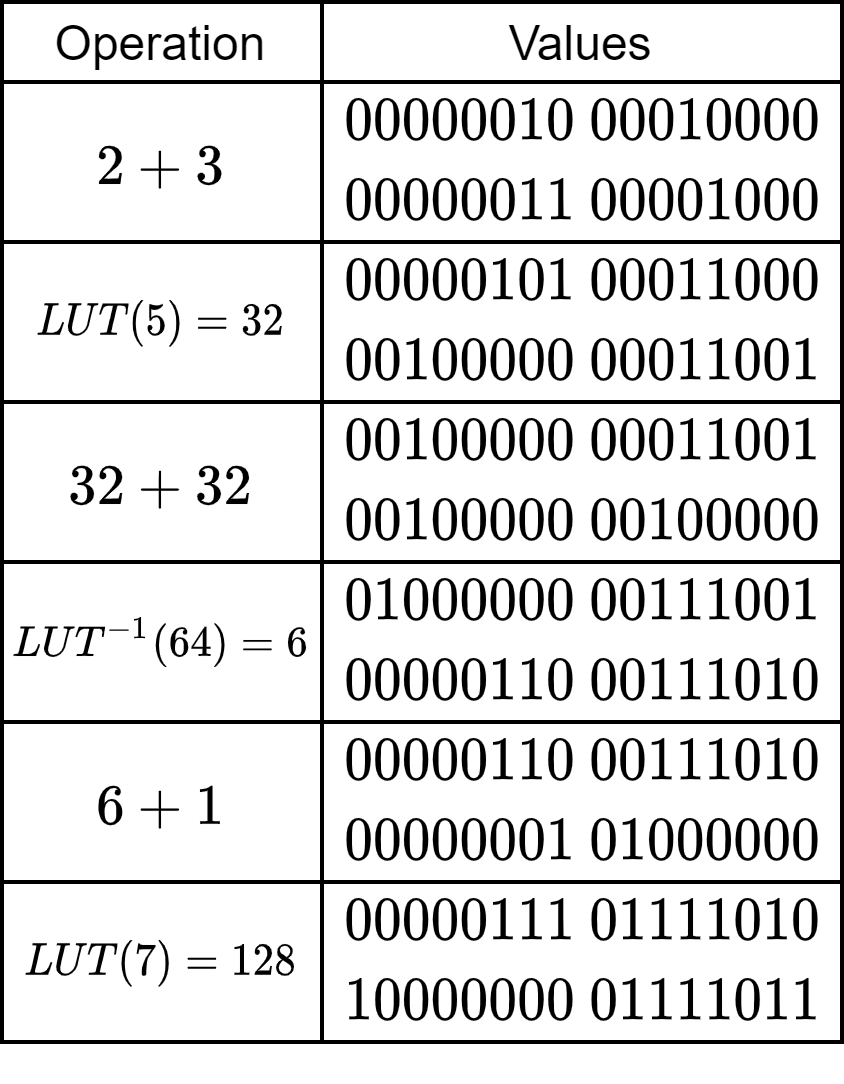}
\caption{Calculation using LUT example}
\label{fig:Calculation LUT example}
\end{figure*}

\begin{lemma}
\label{lemma:multiplication}
For every group of FHE inputs, each having $m$ least significant bits with addition and multiplication (counting) fingerprints ($FP_a,\ FP_m$), and the other $n$ bits used for computation, operated with the addition no overflow-carry blackbox and the exponentiation LUT defined earlier, will result in the correct calculation of the computation and fingerprint parts, or the value 0.
\end{lemma}
\begin{proof}
By Lemma \ref{lemma:addition} the logarithmic addition is correct or set to zero prior to the exponentiations using the LUT. In case the server adds the results of the logarithmic addition and outputs the result without using the LUT, detection of the missing counter in $FP_m$ occurs. As the multiplication fingerprints $FP_m$ reside in the least significant bit part of the fingerprint section itself, even the bounded carry attack attempt using the addition fingerprints will not assist the adversary, as the carry bits could only overflow to the computation part. Any overflow from the computation part effecting the fingerprint section, will also nullify the output. The LUT operations, defined by the delegator, will include all possible monomial and fingerprint values, and will be set to change the $FP_m$ values in a distinct manner.
\end{proof}
\section{Integration and Implementation in Word Granularity}
\label{Integration and Implementation in Word Granularity}
As we deal with multiplications (in logarithmic representation) exponentiation while preserving computation fingerprints, and operating additions, we may want to use blackbox operations that enforce the ordered operations we want to apply. One may use different isolated server/device/hardware restricted to execute only one portion in the above sequence of operation and send the results to the next entity to operate the next operation as required. One opportunity is to use the separated enclave as such computation hardware.

In the previous sections we discussed limitation, that could be enforced by executing different blackboxs as enclaves. The ``Intel SGX'' \cite{costan2016intel} is capable of creating the mentioned enclave, running encrypted code made for a specific program, where the execution of this code is inaccessible for other external processes or entities to interfere. Thus, it will enable us to enforce the use of additions exclusively, eliminating the possibility of any unexpected behavior. 
\\
\\
\noindent {\bf Using existing software libraries.}
\label{Using existing software libraries}
For our example we used Microsoft SEAL library (C\texttt{\#}), where we calculated the following function $F(x,y) = (2\cdot x)+y+3$, where all the constants are in their encrypted form.
Allocating $6$ bits as the size of each part, we have selected (3,2) as the predefined fingerprint values, where $F(3,2)=11$. The fingerprint part was selected as the right (LSB) side, as this choice does not cause overflows. The computation values, generally selected by the computation delegator for it's own use, in this example equal to $(4,7)$, where $F(4,7)=18$.
The addition of constant values was implemented as basic additions along with the requirement of receiving the added value in a doubled form, meaning that for the addition of the constant $3$, the value is received with its shifted counterpart, so the addition will be of a number constructed as $3$, added to the bitwise logically left shifted (indicated as \textit{LSH}) by 6 bits counterpart of it, $LSH_6({3})=192$, where the same was also set for subtraction. The constant multiplication was made as regular multiplication of the constant encrypted value. The input values were $LSH_6$ each, where each result was appended to the requested fingerprint values, in this example, $(3,2)$.

The constructed input was set as $x=LSH_6({4})+3 = 259$, $y=LSH_6({7})+2 = 450$. After calculating $F$ in the form described above, the cleartext result was decrypted as $1163$, of which the least significant part, containing the fingerprint value is extracted using bitwise left and right logical shifting (\textit{LSH}, \textit{RSH}), by calculating $RSH_6(LSH_6({1163}))=11$. Extracting the computation value from the most significant part was by $RSH_6({1163})=18$. Verifying that the former equals our expected fingerprint value of $11$ confirm the calculating, stating that the calculation was verified to have been calculated correctly and that $F(4,7) = 18$.
Out of 3 required additions/subtractions in the original function $F$, the same 3 operations were made, keeping also the original single multiplication as is.
An important note is that the extra addition operations made to the constant values should be computed just once per function definition, and could be reused later and on different input values or servers. This means that the calculation of ($LSH_6({3})+3$) is independent of the (other) input values, and could be calculated offline and used for all further calculations.
The source code of this example can be found in \cite{github}.
\\
\\
\noindent {\bf Supported libraries.}
\label{Supported libraries}
The main libraries tested were Microsoft SEAL \cite{sealcrypto} and IBM HELib \cite{halevi2014algorithms}, both use the BGV/BFV \cite{brakerski2014leveled} and CKKS \cite{cheon2017homomorphic} schemes. Both of those libraries support C\texttt{++} and work on Linux (Ubuntu), where SEAL also works on Windows OS with C\texttt{\#} support. Note, that the CKKS scheme supports floating-point operations. 
The libraries were chosen for their popularity and accessibility, where each of them supports FHE with addition and multiplication. Along with that, HELib focuses on making SIMD operations much faster than previous libraries \cite{halevi2014algorithms}, allowing us to work with vectors of encrypted data, and apply a selected operation on all of the vector elements.
\\
\\
\noindent {\bf Encrypted number representation.}
\label{Encrypted number representation}
As mentioned earlier we consider working on encrypted vector (bit) representation of the data, yet not all schemes have (encrypted) bit access, and on some of the FHE schemes we have an encrypted integer number representation. 
Working with regular FH encrypted numbers involves drastically more computations compared to the vector representation of the data, yet we are still able to operate on it, including the basic blind conditioning defined earlier.

On vector-based variables allowing specific bit access with $w$ bits, the computation complexity of blind conditioning can be represented as a calculation for each of the two possible values of a bit. This will be made for all the bits separately, thus, resulting in a complexity of $O(w)$. On the other hand, with regular FHE numbers (without bit access) the blind conditioning will depend on all of the bit values of the input, combined. This means that we will have a condition for every possible value of the number as a whole, representing a $O(2^w)$ complexity, where $w$ is the total number of bits in the input. Due to that, we refer to both of the representations as equivalent in a qualitative (rather than quantitative) manner, as they have the same output capabilities.
\\
\\
\noindent {\bf Fields.}
\label{Fields}
While limiting the number of operations possible to eliminate the chance of overflowing each of the part limits, we can elaborate on a different way to \textit{reinitialize} our verification layer, using modulo and finite fields. This can be implemented in a simple form by providing the modulo value as an encrypted number from the delegator, used by the server, to subtract any whole multiplication of the field. An overflow can be detected using a {\em blind if} (algorithm \ref{alg:Blind}) on a bit indicating it, and the subtracted result will be used for further calculation. The provided encrypted value from the delegator could also possibly include a module value to be subtracted from the fingerprint part, yet the two parts are not dependent, and the fingerprint section can be set to not be affected by the subtraction.

\begin{lemma}
\label{lemma:arithmetic circuit}
Any calculation of an arithmetic circuit, representing a polynomial, having FHE inputs that are infused with computation fingerprints following our scheme, can be successfully verified by the delegator.
\end{lemma}
\begin{proof}
By Lemma \ref{lemma:addition} and Lemma \ref{lemma:multiplication}, the addition and multiplication of encrypted values, infused with fingerprints, will result in the correct calculation of the computation and fingerprint parts, or the value 0. Any polynomial is the addition of monomials, where each monomial is the multiplication of values. Thus, we can represent and calculate any polynomial by our scheme, therefore, we can calculate any arithmetic circuit. Receiving the expected fingerprint value will verify the correct execution of the arithmetic circuit, and any different fingerprint result will deem otherwise.
\end{proof}
\section{Restrictive SIMD as a vector with Fingerprints}
\label{Restrictive SIMD as a Vector with Fingerprints}
We move on to the possibility of incorporating the \textit{computational fingerprint} into SIMD manipulated data structures, such as vectors. Starting with vectors of integer elements, we proceed to vectors with only fingerprint integer values, enabling floating point arithmetic for the delegator.
\\
\\
\noindent {\bf Background on floating point operation on data in SIMD.}
\label{Background on floating point}
The primary FHE floating point scheme is the CKKS \cite{cheon2017homomorphic}, which works by encoding a vector of plaintext values to a polynomial. Each coefficient of the polynomial is originated from the corresponding vector element, which is projected to a different field, and then multiplied by a scaling factor to control the rounding error of the floating point number. This polynomial is defined by treating the vector values as the image for a selected set of roots, where the degree of the resulted polynomial will equal the vector length. In addition to that, there is a special technique called rescaling, that consistently maintains the decryption structure small enough compared to the ciphertext modulus, as a kind of modulus switch operation \cite{brakerski2014leveled}.

The rescaling operation reduces the size of the ciphertext modulus, while maintaining a valid encryption and almost preserving the precision of the plaintext, using primes selected as part of the encryption parameters, so at each rescaling execution the ciphertext modulo is divided by one of the chosen primes. It is important to note that as in modulus switching, the number of primes limits the amount of possible rescaling operations, and thus limits the possible depth (level) of the computation. This limit will enable us to restrict overflow attempts caused by redundant SIMD operations.

FHE operations are obtained by manipulating each monomial of the polynomial accordingly. This succinct polynomial representation is well suited for SIMD calculations, where each operation will be acted upon in the same manner, on all polynomial coefficients. Still, we can regard each entry of the SIMD vector as a totally independent value.
\\
\\
\noindent {\bf Adversarial capabilities in SIMD environment.}
\label{Adversarial capabilities SIMD}
Although instructed otherwise, the adversary can compute a different arithmetic circuit (or general computation, which is not necessarily defined in a finite field) than the one the delegator instructs the server to use. Still, the server is bounded to perform the exact same sequence of operations on each entry of the data vector, and cannot execute different programs for each such entry. Moreover, there is isolation between vector entries, with no mutual influence.

We choose to restrict our vector components, both fingerprint and computation elements, to integers, as using floating point may result in rounding errors, possibly nullifying the effect of (very) small fingerprint numbers not in the scale of (much) bigger numbers when, say, they are added to each other, and this will result in approximating the small number to 0 (see e.g., \cite{goldberg1991every}).
\\
\\
\noindent {\bf Integers in SIMD.}
\label{Integers in SIMD}
Using the floating point CKKS scheme described earlier, we are able to manipulate vectors containing different elements in an efficient SIMD manner. We will adapt the \textit{integer fingerprints} scheme to comply with this new state, where separate elements reside in a vector. As in the basic \textit{integer fingerprints} scheme, the fingerprint element, manipulated by the requested arithmetic circuit from the delegator, will result in a value the delegator can better verify.

By limiting our scheme to work with only integer elements, we also restrict divisions, as this impairs any possibility of creating floating point values while executing a calculation.
Regarding possible overflows, discussed in the previous sections, any overflow of the vector will affect all of its elements in the same manner, without mutual influence between different elements. This means that overflows might just create new (overflowed) values or return the vector to its original state, meaning that in this scheme, we do not need to require a blackbox that identify and act upon overflows.

Note that the usage of integers in SIMD allows us to skip the log representation conversion (using the LUT) described earlier for the word granularity case, making our scheme even more efficient. In addition to that, for the same arithmetic circuit, we can verify multiply isolated inputs in a vector, using a single fingerprint value in it.
\\
\\
\noindent {\bf Fingerprints in floating point vectors.}
\label{Fingerprints in floating}
SIMD supports floating point variables and operations, still, we prefer the fingerprints to be restricted to integer and to avoid divisions and subtractions altogether. Additionally, we avoid adding any constant floating point values to the input, as those also affect the fingerprint element. The reasons for restricting the fingerprints to be integer are related to the special rounding operations applied in the scope of float representations. Nevertheless, we allow floating point values to be used as the computation elements.

Avoiding divisions and subtractions includes limiting all fingerprint values to be a positive integer (as discussed above, greater than 1), to prohibit subtraction by an addition of negative values, or division as the multiplication of fractures.

As for avoiding divisions and subtractions, we can take, for example, a floating point computation element with a value of $0.0001$, added to a $10000000000000.0$ and then immediately subtracted by the same $10000000000000.0$. The addition and subtraction of the floating point element may be approximated to $0$, yet the corresponding integer fingerprint elements that maybe, say 7 and 4, will return to the original value of 7, which makes the floating point rounding invisible to the verifying delegator, thus, we might verify a wrong computation result as the right result.

Even if we use only integers as fingerprints, changing the ordering of operation, might be used to maliciously create the correct fingerprint result. The possibility of this manipulation will be thoroughly discussed in the following section.

An implementation example of calculating the polynomial $F(x,y) = (((2 \cdot x)+1.5) \cdot (y \cdot 3)) + 0.1$, with SIMD manipulation of a vector containing integer fingerprint element and floating point computational values, using Microsoft SEAL, can be found in \cite{github}.
\\
\\
\noindent
{\bf Computation fingerprint circuit.}
We would like to mention that one may also want to trace the order of operations made (e.g., first adding small numbers to be non negligible before adding them to a very large number, a serial addition of each one of them to the big number results in an approximation to $0$). We suggest extending the tracing by having the blackbox produce (possibly during the computation or as a final result) a symbolic (with or without FHE known computed values) representation of the computation steps, in a form of an arithmetic circuit.

\begin{theorem}
Under the surveyed restrictions, computation fingerprint circuit verifies the correctness of the computation made in SIMD.
\end{theorem}
\section{Conclusion Remarks}
\label{Conclusion Remarks}
We presented a new approach to use computation fingerprints within FHE values to check a delegated computation result. Our word granularity and SIMD solutions are designed to check the final result of the computation, without any significant additional overhead.

Thus, having FHE inputs and a desired arithmetic circuit, any entity following our scheme can delegate the required computation to any other party, have a with high probability guaranteed that the computation was made as requested, and accomplish this procedure by computing the calculation exactly once, with no redundancy.

The SIMD \cite{smart2014fully} optimization that many FHE libraries provide allows performing an individual computation, like addition or multiplication, on a vector of encrypted elements. This unique optimization will give us a substantial advantage over the word granularity FHE schemes, as the scope of the fingerprint method may be extended to the case of non-integer values, including floating point calculations (the fingerprint can still be restricted to integers, but the actual variables can be also beyond integers). Note that the restriction of using only SIMD server may yield a less efficient solution than the computer word solution, as it requires to compute the function more than once, either in parallel or sequential manner. Still, when a computation should be executed in parallel in the first place, or when the function is defined over non integer values (e.g., floating point) and is beyond the computation capabilities of an arithmetic circuit over a finite field, the SIMD setting can serve us better.

Verifiable computing may be used on almost any of the many FHE applications, as Homomorphic encryption in its essence is oriented toward delegating data to be operated in untrusted environments \cite{armknecht2015guide}, thus, there is great benefit in verifying and providing assurance for those type of computations.

\end{document}